\documentclass[aps,pra,twocolumn,superscriptaddress,nofootinbib]{revtex4-2}

 \usepackage{amsthm}\newtheorem{theorem}{Theorem}
\usepackage[braket, qm]{qcircuit} 
\usepackage{graphicx}
\usepackage{amssymb}
\usepackage{indentfirst}
\usepackage{amsmath}
\usepackage{color,xcolor}
\usepackage{subfigure}
\def \red #1 {\textcolor{red}{#1}}
 \usepackage{verbatim}
\usepackage{hyperref}
\hypersetup{
    colorlinks=true,
    linkcolor=blue,
    citecolor=blue,
    filecolor=blue,
    urlcolor=blue
}

\def\tr{\mathop{\mathrm{tr}}}

\theoremstyle{definition}

\newtheorem{lemma}{Lemma}
\newtheorem{remark}{Remark}

\usepackage[noend]{algpseudocode}
\usepackage{algorithm,algorithmicx}

\algrenewcommand\alglinenumber[1]{\sf\scriptsize\color{blue}{#1}}
\algrenewcommand\algorithmicrequire{\textbf{Input:}}
\algrenewcommand\algorithmicensure{\textbf{Output:}}

\begin{document}

\title{Direct reconstruction of the quantum density matrix elements with classical shadow tomography}

\author{Yu Wang}
\email[]{wangyu@bimsa.cn}
\affiliation{Beijing Institute of Mathematical Sciences and Applications, Beijing 101408, China}

%\author{Jinfeng Zeng}
%\email[]{}

%\affiliation{Beijing Academy of Quantum Information And Science}

\begin{abstract}
We introduce a direct estimation framework for reconstructing multiple density matrix elements of an unknown quantum state using classical shadow tomography. Traditional direct measurement protocols (DMPs), while effective for individual elements, suffer from poor scalability due to post-selection losses and the need for element-specific measurement configurations. In contrast, our method—DMP-ST—leverages random Clifford or biased mutually unbiased basis measurements to enable global estimation: a single dataset suffices to estimate arbitrary off-diagonal entries with high accuracy. 
We prove that estimating \(K\) off-diagonal matrix elements up to additive error \(\epsilon\) requires only \(\mathcal{O}(\log K/\epsilon^2)\) samples, achieving exponential improvement over conventional DMPs. 
The number of required measurement configurations can also be exponentially reduced for large K. 
When extended to full state tomography, DMP-ST attains trace distance error \(\le \epsilon\) with sample complexity \(\mathcal{O}(d^3 \log d/\epsilon^2)\), which is closed to the optimal scaling for single-copy measurements. Moreover, biased MUB measurements reduce sample complexity by a constant factor than random Clifford measurements. 
This work provides both theoretical guarantees and explicit protocols for efficient, entrywise quantum state reconstruction. It significantly advances the practicality of direct tomography, especially for high-dimensional systems and near-term quantum platforms.
 
\end{abstract}

\maketitle

\section{Introduction}

Quantum states are fundamentally described by density matrices, which encapsulate all statistical properties of a system. In quantum mechanics, any strong projective measurement irreversibly collapses the measureed quantum state, discarding valuable information about the original system. Consequently, characterizing an unknown quantum state often requires repeated preparation and measurement, making sample efficiency a central concern in quantum information science~\cite{paris2004quantum,nielsen2010quantum}.

A key goal in this context is the estimation of density matrix elements, especially off-diagonal terms that encode quantum coherence and entanglement~\cite{friis2019entanglement,streltsov2017colloquium}. Traditional quantum state tomography (QST) reconstructs the full density matrix via overcomplete measurements and convex optimization~\cite{gross2010quantum,flammia2012quantum}, incurring sample complexity that grows rapidly with the system dimension \( d \). While accurate, such methods are resource-intensive and often impractical for near-term quantum devices.

To address this, direct measurement protocols (DMPs)~\cite{lundeen2011direct,lundeen2012procedure,salvail2013full,bamber2014observing,shi2015scan,thekkadath2016direct,bolduc2016direct,vallone2016strong,calderaro2018direct,pan2019direct,zhang2020direct,xu2024resource} have been proposed to estimate specific entries of the density matrix without full tomography. These methods typically employ weak or strong measurements coupled with ancillary probes and targeted unitary operations, enabling constant-cost access to individual matrix elements. However, DMPs face two fundamental limitations: (1) many measurement outcomes for a single off-diagonal element $\rho_{jk}$ are discarded due to post-selection, reducing sample efficiency; and (2) each new off-diagonal element generally requires a distinct measurement configuration, making it impossible to reuse data from the past measurement results.

Recently, classical shadow tomography~\cite{huang2020predicting} (building on prior work by Aaronson \cite{aaronson2018shadow}) has emerged as a powerful tool for predicting multiple properties of an unknown quantum state from a shared dataset. By applying randomized measurements followed by classical post-processing, it enables efficient estimation of expectation values for a large set of observables with sample complexity scaling only logarithmically in the number of targets. Shadow-based methods have been successfully applied to a broad range of quantum information tasks, including fidelity estimation, entropy and Hamiltonian learning~\cite{kokail2021entanglement,hadfield2022measurements}, error mitigation~\cite{seif2023shadow,jnane2024quantum,zhao2024group}, and quantum machine learning~\cite{huang2022provably,jerbi2024shadows}, among others. 

In this work, we develop a hybrid framework—direct measurement protocols via shadow tomography (DMP-ST)—that leverages classical shadow techniques to directly estimate density matrix elements. Unlike conventional DMPs, DMP-ST allows all measurement outcomes to be reused across different target elements, offering significantly improved sample efficiency and reduced experimental overhead. Our framework supports both random Clifford measurements and biased mutually unbiased bases (MUBs), with the latter providing further improvements in sampling cost and hardware simplicity.

We prove that estimating \(K\) off-diagonal density matrix elements to additive error \(\epsilon\) requires only \(\mathcal{O}(\log K / \epsilon^2)\) samples, improving exponentially over the \(\mathcal{O}(K/\epsilon^2)\) cost of traditional DMPs in the worst case. Furthermore, we show that reconstructing all \(d^2\) matrix elements with per-entry error \(\epsilon_1 = \mathcal{O}(\epsilon/d^{3/2})\) guarantees trace distance error \(\le \epsilon\), achieving full quantum state tomography with total sample complexity
\[
\mathcal{O}\left( \frac{d^3 \log d}{\epsilon^2} \right),
\]
which matches the best-known upper bound for single-copy measurement schemes using approximate 4-designs~\cite{kueng2017low}, and approaches the tight lower bound of \(\Omega\left( \frac{d^3}{\epsilon^2} \right)\)  established by~\cite{chen2023does}.

In contrast to standard QST techniques—such as linear inversion, semidefinite programming, or maximum likelihood estimation—our DMP-ST protocol enables direct, entrywise estimation using lightweight post-processing. 
%This not only reduces classical computational cost but also improves experimental feasibility for large-scale quantum systems. 
The DMP-ST framework thus offers a new, scalable route toward efficient quantum state reconstruction, balancing theoretical guarantees with practical implementation.

\section{Sample complexity of direct measurement protocols}

%\paragraph*{Diagonal elements.} 
\subsection{Diagonal Elements Estimation}
In quantum state characterization, the diagonal elements \(\rho_{jj}\) of a density matrix \(\rho\) correspond to the outcome probabilities of performing projective measurements in the computational basis. Specifically, according to Born’s rule, the probability of obtaining outcome \(j\) when measuring \(\rho\) in the basis \(\{|j\rangle\}\) is exactly \(\rho_{jj}\).  

As such, these diagonal entries can be directly estimated by repeating the computational basis measurement and recording the relative frequencies of each outcome. While this estimation is operationally simple, it remains important to analyze the number of samples required to achieve a desired accuracy and confidence. In particular, we are interested in the sample complexity for estimating either a single diagonal entry or multiple entries simultaneously with additive error \(\epsilon\) and failure probability at most \(\delta\). The following lemma quantifies this complexity.

\begin{theorem}
\label{lemma:diagonal_sample_complexity}
Let \(\rho\) be a \(d \times d\) density matrix. To estimate its diagonal elements with additive error \(\epsilon\) and failure probability at most \(\delta\), the required number of samples is:
\begin{enumerate}
    \item For estimating a single diagonal element \(\rho_{jj}\):
    \begin{equation}
        N \geq \frac{\ln(2/\delta)}{2\epsilon^2}.
    \end{equation}
    \item For estimating \(K\) different  diagonal elements \(\{\rho_{j_1 j_1}, \dots, \rho_{j_K j_K}\}\) simultaneously:
    \begin{equation}
        N \geq \frac{\ln(2K/\delta)}{2\epsilon^2}.
    \end{equation}
\end{enumerate}
\end{theorem}

\begin{proof}
To estimate diagonal elements of \(\rho\), we perform projective measurements in the computational basis \(\{|j\rangle\langle j|\}_{j=0}^{d-1}\).  
According to Born's rule, the probability of obtaining outcome \(j\) is  
\[
p_j = \mathrm{tr}(\rho |j\rangle\langle j|) = \rho_{jj}.
\]

Suppose we perform \(N\) independent measurements totally. Let \(n_j\) denote the number of times outcome \(j\) is observed. Then the empirical frequency  
\[
\hat{p}_j = \frac{n_j}{N}
\]
serves as an unbiased estimator of \(\rho_{jj}\).

\textit{(1) Single diagonal element.}  
By Hoeffding’s inequality, the probability that \(\hat{p}_j\) deviates from \(\rho_{jj}\) by more than \(\epsilon\) is bounded as
\[
\mathbb{P}(|\hat{p}_j - \rho_{jj}| \geq \epsilon) \leq 2 \exp(-2N\epsilon^2).
\]
To ensure this failure probability is at most \(\delta\), it suffices to require
\[
2 \exp(-2N\epsilon^2) \leq \delta,
\]
which yields
\[
N \geq \frac{\ln(2/\delta)}{2\epsilon^2}.
\]

\textit{(2) Multiple diagonal elements.}  
Now consider simultaneously estimating \(K\) different diagonal elements \(\rho_{j_1 j_1}, \dots, \rho_{j_K j_K}\).  
We require that all estimates \(\hat{\rho}_{j_t j_t}\) deviate from the true values by at most \(\epsilon\), with overall failure probability at most \(\delta\).

By the union bound,
\begin{align*}
\mathbb{P}\Big(\exists\, t,\,
  \big|\hat{\rho}_{j_t j_t} - \rho_{j_t j_t}\big| \geq \epsilon\Big)
&\leq \sum_{t=1}^K \mathbb{P}\left(\big|\hat{\rho}_{j_t j_t} - \rho_{j_t j_t}\big| \geq \epsilon\right) \\
&\leq 2K \exp(-2N\epsilon^2).
\end{align*}
To ensure the total failure probability is at most \(\delta\), we require
\[
2K \exp(-2N\epsilon^2) \leq \delta.
\]

which implies
\[
N \geq \frac{\ln(2K/\delta)}{2\epsilon^2}.
\]
\end{proof}

\begin{remark}
    
Thus, estimating a single diagonal element requires a sample complexity scaling as \(\mathcal{O}(1/\epsilon^2)\), whereas estimating \(K\) diagonal elements introduces only a logarithmic overhead, resulting in a sample complexity of \(\mathcal{O}(\log K / \epsilon^2)\).  

At first glance, this may seem counterintuitive: one might expect that estimating \(K\) independent quantities would require sample complexity scaling linearly with \(K\). However, this intuition overlooks an important structural advantage—all \(K\) quantities are estimated from the same set of measurement outcomes. The union bound is then applied to control the worst-case failure probability across all \(K\) estimates, leading to the appearance of the \(\log K\) factor.  

This is a compelling illustration of the statistical efficiency of global estimation: by leveraging shared measurement data and controlling the total failure probability collectively, one can achieve accurate estimation of multiple quantities with only a modest increase in resources. The result is both theoretically sound and practically powerful.

\end{remark}

\subsection{Off-diagonal Elements Estimation}
While diagonal elements of a density matrix can be efficiently estimated using simple projective measurements in the computational basis, estimating off-diagonal elements typically requires more specialized procedures.  
These off-diagonal entries, which encode quantum coherences between different basis states, play a central role in characterizing key quantum phenomena such as entanglement~\cite{friis2019entanglement,horodecki2009quantum} and decoherence~\cite{streltsov2017colloquium,ringbauer2018certification}.  
This makes their accurate estimation not only technically challenging but also physically essential.  
Direct measurement protocols are specifically designed to address this challenge~\cite{lundeen2011direct,lundeen2012procedure,salvail2013full,bamber2014observing,shi2015scan,thekkadath2016direct,bolduc2016direct,vallone2016strong,calderaro2018direct,pan2019direct,zhang2020direct,xu2024resource,feng2021direct,wang2024direct}. 
Unlike conventional quantum state tomography, which typically requires at least \(\mathcal{O}(d)\) different measurement operators and substantial classical post-processing to reconstruct the entire density matrix, DMPs enable the direct estimation of arbitrary single off-diagonal element using only a constant number of measurement settings, which are equivalent to a constant types of unitary transformations followed by projective measurement in the computational basis. 
Importantly, since the density matrix \(\rho\) is Hermitian, we have \(\rho_{jk} = \rho_{kj}^*\).  
Therefore, matrix elements at symmetric positions are not independent and can be treated as equivalent, meaning it suffices to estimate only one of them.

However, aside from such symmetric pairs, estimating different off-diagonal elements often requires entirely different measurement setups, which breaks the possibility of reusing the same measurement resources. This issue arises in all known direct measurement protocols. For example:

In phase-shifting schemes \cite{feng2021direct}, each off-diagonal element corresponds to a distinct phase-sensitive interference pattern. Therefore, estimating a new target element—regardless of its position—requires a new measurement configuration, with no possibility of reusing previous settings. As a result, the total number of distinct measurement settings needed to access all density matrix elements scales as \( \mathcal{O}(d^2) \).

In strong-measurement-based schemes \cite{calderaro2018direct}, each off-diagonal element \( \rho_{jk} \) is extracted by coupling the main system with two ancillary qubits through a specific interaction that depends on the indices \( j \) and \( k \). Crucially, when the new target element \( \rho_{j'k'} \) does not lie in the same row or column as the previously measured element, the coupling configuration between the main system and the first ancilla qubit must be modified. The total number of measurement settings required for full density matrix estimation is \( \mathcal{O}(d) \).

In direct measurement schemes based on dense dual bases \cite{wang2024direct}, each matrix element \( \rho_{jk} \) is associated with a superposition state such as \( |j\rangle + |k\rangle \) and \( |j\rangle + i|k\rangle \). When new elements correspond to superpositions $j', k'$ belonging to different orthonormal bases in the DDB family, the measurement basis must be switched accordingly. The total number of measurement settings needed to reconstruct all density matrix elements also scales as \( \mathcal{O}(d) \).

In all these cases, the estimation procedures for different matrix elements are inherently incompatible. As a result, estimating \( K \) arbitrary off-diagonal elements would require \( \mathcal{O}(K) \) different measurement configurations, leading to a sample complexity that scales as \( \mathcal{O}\left(\frac{K \ln \frac{1}{\delta}}{\epsilon^2 }\right) \) in the worst case. Each measurement setting must be repeated \( \mathcal{O}(\ln \frac{1}{\delta}/(\epsilon^2 )) \) times to achieve additive error \( \epsilon \) in estimating probabilities of the form \( \operatorname{tr}(U \rho U^{\dagger} |k\rangle\langle k|) \), where \( U \) is the unitary operation specified by the direct measurement protocol, and \( |k\rangle \) is the post-selected basis state measured on the main system. 

\section{Estimating multiple density matrix elements with classical shadow tomography}

In contrast, classical shadow tomography \cite{huang2020predicting} supports global estimation by design: a single batch of measurement data can be reused to estimate a large collection of observables \( O_1, O_2, \dots, O_K \), with sample complexity of \( \mathcal{O}(\log K / \epsilon^2) \). However, the class of observables that can be efficiently estimated depends on the specific structure related to the choice of randomized measurements. 
 
This exponential efficiency is analogous to the diagonal case, where repeated measurements in the computational basis allow simultaneous estimation of multiple diagonal entries. Similarly, shadow tomography exploits the global reuse of measurement outcomes to estimate many observables with collective error guarantees, yielding a sample complexity that scales only logarithmically with the number of targets $K$. 

 In the following, we introduce the classical shadow framework and demonstrate how it can be tailored to efficiently and directly estimate individual matrix elements. 

%\paragraph*{Classical Shadow Tomography Framework.}
\subsection{Classical Shadow Tomography Framework}
Specifically, classical shadow tomography \cite{huang2020predicting} is an efficient protocol for predicting many properties of an unknown quantum state \(\rho\). The process involves randomly sampling a unitary transformation \(U_k\) from an informationally complete ensemble \(\mathcal{U} = \{U_k\}\), evolving the state \(\rho \to U_k \rho U_k^\dagger\), and performing projective measurements in the computational basis. Each measurement outcome \(j\) corresponds to a classical snapshot, reconstructed via the inverse measurement channel \(\mathcal{M}^{-1}\):
\[
\hat{\rho}_j = \mathcal{M}^{-1}(U_k^\dagger |j\rangle \langle j| U_k).
\]
Repeating the procedure \( N \) times generates a set of \( N \) independent snapshots, \( S(\rho, N) = \{\hat{\rho}_1, \hat{\rho}_2, \ldots, \hat{\rho}_N\} \), referred to as the classical shadow of \( \rho \). The classical shadow enables efficient prediction of properties such as \( \{\mbox{tr}(O_i \rho)\} \) for observables \( \{O_i\} \). By splitting \( S(\rho, N) \) into equally sized chunks and employing the median-of-means estimation technique, we can mitigate the effects of outliers while ensuring rigorous performance guarantees.

For \(K\) target observables \(\{O_1, \dots, O_K\}\), the sample complexity satisfies:
\begin{equation}\label{eq:sample-general}
N = O\left(\frac{\log (K/\delta)}{\epsilon^2} \max_i \|O_i\|_{\text{shadow}}^2 \right),    
\end{equation}
where \(\|O_i\|_{\text{shadow}}\) depends on the choice of random unitary ensemble and the target observables. This guarantees an additive error \(\epsilon\) in the predicted values with high confidence $1-\delta$.

It is well known that any off-diagonal element of a density matrix can be estimated using only four rank-1 projective measurements \cite{Caves2002}. The following lemma provides the explicit reconstruction formula:

\begin{lemma}
\label{lemma:off_diagonal_reconstruction}
Let \(\rho = \sum_{j,k} \rho_{jk} |j\rangle\langle k|\) be a \(d \times d\) density matrix.  
For any \(j \neq k\), the off-diagonal element \(\rho_{jk}\) can be reconstructed from the expectation values of four rank-1 observables: \(|j\rangle\langle j|\), \(|k\rangle\langle k|\), \(|\phi_{jk}\rangle\langle \phi_{jk}|\), and \(|\psi_{jk}\rangle\langle \psi_{jk}|\), where:
\begin{equation}
   |\phi_{jk}\rangle = \frac{|j\rangle + |k\rangle}{\sqrt{2}}, \quad
   |\psi_{jk}\rangle = \frac{|j\rangle + i|k\rangle}{\sqrt{2}}.
\end{equation}
\end{lemma}

\begin{proof}
Since \(\rho_{jk} = \rho_{kj}^*\), we can reconstruct its real and imaginary parts separately using projective measurements. The expectation values yield:
\begin{equation}\label{equ:r+i}
    \begin{aligned}
    \mathrm{Re}(\rho_{jk}) &= \mathrm{tr}(\rho\, |\phi_{jk}\rangle\langle \phi_{jk}|) - \frac{1}{2}(\rho_{jj} + \rho_{kk}), \\
    \mathrm{Im}(\rho_{jk}) &= -\mathrm{tr}(\rho\, |\psi_{jk}\rangle\langle \psi_{jk}|) + \frac{1}{2}(\rho_{jj} + \rho_{kk}).
\end{aligned}
\end{equation}

Thus, the full complex matrix element is recovered by:
\begin{equation}\label{equ:total}
    \rho_{jk} = \mathrm{Re}(\rho_{jk}) + i\, \mathrm{Im}(\rho_{jk}).
\end{equation}
\end{proof}

 We define a set of rank-1  observables \(\mathcal{O}\), consisting of \(d(d-1)\) rank-1 projection operators, as follows:  
\begin{equation}\label{eq:set-observable}
   \mathcal{O} = \left\{ |\phi_{jk}\rangle \langle \phi_{jk}|, |\psi_{jk}\rangle \langle \psi_{jk}| \;\middle|\; 0 \leq j < k \leq d-1 \right\}, 
\end{equation}

\subsection{Estimating Off-diagonal Elements via Classical Shadow Tomography}
\begin{lemma}  
\label{lemma:shadow_sampling}  
To estimate the expectation values of \(K\) rank-1 observables in the set \(\mathcal{O}\) (Eq.~(\ref{eq:set-observable})), such that each estimated value \(\hat{o}_i\) satisfies:  
\[ 
|\hat{o}_i - \mathrm{tr}(\rho O_i)| \leq \epsilon \quad \text{for all } i \in \{1, \dots, K\}, 
\]  
with a success probability of at least \(1 - \delta\), the worst-case sample complexity within the classical shadow tomography framework is as follows:   

\begin{itemize}
    \item \textbf{Random Clifford Measurements:}  
    \begin{equation}\label{eq:sample-clifford}
        N_{\mathrm{Clifford}} = O\left( \frac{\log (K / \delta)}{\epsilon^2} \right).
    \end{equation}
\item \textbf{Biased MUB Measurements:}  
We define a biased sampling distribution \(\mathcal{D}\) over the complete set of mutually unbiased bases (MUBs) \(\{ \mathcal{B}_0, \mathcal{B}_1, \dots, \mathcal{B}_{2^n} \}\), where \(\mathcal{B}_0\) denotes the computational basis. The distribution is defined as:
\begin{equation}\label{eq:biasedprobability}
\Pr_{\mathcal{D}}[\mathcal{B}_j] =
\begin{cases}
\displaystyle \frac{1}{2}, & \text{if } j = 0 \quad \text{(computational basis)}, \\\\[6pt]
\displaystyle \frac{1}{2^{n+1}}, & \text{if } 1 \le j \le 2^n \quad \text{(nontrivial MUBs)}.
\end{cases}
\end{equation}

This biased sampling slightly  reduces the total number of measurements required. In particular, the required number of samples for the worst case  satisfies:
\begin{equation}
    N_{\mathrm{MUB}} \approx \frac{1}{2} N_{\mathrm{Clifford}}.
\end{equation}

\end{itemize}
\end{lemma}

\begin{proof}
For observables $O_i$ in \(\mathcal{O}\), classical shadow tomography based on random Clifford measurements \cite{huang2020predicting} provides a reliable framework with a reconstruction channel:
\[
\mathcal{M}^{-1}(X) = (2^n + 1)X - \mathrm{tr}(X)I,
\]
with shadow norm bounded by:
\[
\max_i \|O_i\|_{\text{shadow}}^2 \le \max_i 3 \cdot \mathrm{tr}([O_i - I/2^n]^2) = 3(1 - 1/2^n).
\]
This leads to the sample complexity in Eq.~(\ref{eq:sample-clifford}).

 Since each observable in \(\mathcal{O}\) is a rank-1 projector supported on only two computational basis vectors, it contains exactly \(t = 2\) nonzero entries and is therefore MUB-sparse in the sense of Theorem~2 in~\cite{wang2024classical}.

We follow the biased sampling strategy described in Eq.~(\ref{eq:biasedprobability}) for MUB-sparse observables. 
 The reconstruction channel is:
\[
\mathcal{M}_b^{-1}(X) = 2^{n+1} \left[ X - \frac{2^n - 1}{2^n} \sum_{k=0}^{2^n - 1} \mathrm{tr}(X P_{0k}) P_{0k} \right] - \frac{\mathrm{tr}(X) I}{2^n},
\]
where \(P_{0k} = |k\rangle\langle k|\) denotes the computational basis projectors.

Although Theorem~2 in \cite{wang2024classical} provides a general bound on the shadow norm as:
\[
\max_i \|O_i\|_{\text{shadow}}^2 \le t^2 = 4,
\]
this bound is based on the worst-case setting for general \( t \)-sparse observables. However, in our case, each observable in \( \mathcal{O} \) is supported only on two computational basis states and all nonzero amplitude coefficients are of equal magnitude. This specific structure allows us to derive a tighter shadow norm bound.

For each \(O \in \mathcal{O}\), define the traceless part as \( O_0 = O - \frac{\mathrm{tr}(O)}{2^n} I \). The variance under biased MUB measurement is given by Eq.~(B12) in \cite{wang2024classical}:
\begin{equation}
\begin{aligned}
\|O_0\|_{\sigma,b}^2 =
&\ 2^{n+1} \sum_{j=1}^{2^n} \sum_{k=0}^{2^n - 1} \operatorname{tr}^2(O_0 P_{jk}) \operatorname{tr}(\sigma P_{jk}) \\
&\quad + 2 \sum_{k=0}^{2^n - 1} \operatorname{tr}^2(O_0 P_{0k}) \operatorname{tr}(\sigma P_{0k}),
\end{aligned}
\end{equation}
where \( P_{jk} = U_j^\dagger |k\rangle\langle k| U_j \) for nontrivial MUB unitaries \( j = 1, \dots, 2^n - 1 \), and \( P_{0k} = |k\rangle\langle k| \) for the computational basis.

Each observable \( O=|\phi\rangle\langle\phi| \in \mathcal{O} \) is a rank-1 projector supported on the subspace spanned by two computational basis vectors \( |\ell_1\rangle, |\ell_2\rangle \), and takes the form:
\[
|\phi\rangle \in \left\{
\frac{1}{\sqrt{2}}(|\ell_1\rangle \pm |\ell_2\rangle),\quad
\frac{1}{\sqrt{2}}(|\ell_1\rangle \pm i |\ell_2\rangle)
\right\}.
\]

To analyze the variance for the computational basis, we observe that for each \( k \), \( \operatorname{tr}(O|k\rangle\langle k|) \le 1/2 \), hence:
\[
\operatorname{tr}(O_0 |k\rangle\langle k|) = \operatorname{tr}(O|k\rangle\langle k|) - 1/2^n \le 1/2 - 1/2^n,
\]
so:
\[
\operatorname{tr}^2(O_0 |k\rangle\langle k|) \le (1/2 - 1/2^n)^2.
\]
Summing over all \( k \) and using \( \sum_k \operatorname{tr}(\sigma P_{0k}) = 1 \), the total contribution is:
\[
2 \sum_k \operatorname{tr}^2(O_0 P_{0k}) \operatorname{tr}(\sigma P_{0k}) \le \frac{1}{2} + O(1/2^n).
\]

By Result~6(i) of~\cite{yu2023efficient}, for any fixed nontrivial MUB unitary \( U_j^{\dagger} \), the rescaled amplitudes \( \sqrt{2^n}\langle \ell | U_j^\dagger |k\rangle \), for fixed \( \ell \ne 0 \) and varying \( k \), are either all in \( \{ \pm 1 \} \) or all in \( \{ \pm i \} \), and are evenly distributed: half of the values are \( +1 \) (or \( +i \)) and the other half are \( -1 \) (or \( -i \)).

Moreover, if we fix \( \ell \ne 0 \) and vary both the MUB index \( j = 1, \dots, 2^n \) and the basis index \( k = 0, \dots, 2^n - 1 \), then the complete set of values
\(
\left\{ \sqrt{2^n} \langle \ell | U_j^\dagger | k \rangle \right\}
\)
consists of \(4^n\) complex numbers. According to Result~6(ii) in~\cite{yu2023efficient}, the values \( \{ \pm 1, \pm i \} \) appear with equal frequency, each occurring in exactly one-quarter of the total entries.

We now compute the quantity \( \operatorname{tr}^2(O_0 P_{jk}) \), where the traceless observable is defined as
\[
O_0 = |\phi\rangle\langle \phi| - \frac{I}{2^n}.
\]

 By examining the relative phase between the \( \ell_1 \) and \( \ell_2 \) components of the MUB state \( U_j^\dagger |k\rangle \), the nontrivial unitaries \( \{ U_j \}_{j=1}^{2^n} \) can be partitioned into two disjoint sets based on their interference behavior with the state \( |\phi\rangle \in \mathcal{O} \):

\begin{itemize}
  \item \textbf{\( \mathcal{U}_{O,1} \)}: The relative phase between the \( \ell_1 \)-th and \( \ell_2 \)-th components of \( U_j^\dagger |k\rangle \) matches that of the target state \( |\phi\rangle \), for all \( k \). As a result, the contributions to \( \operatorname{tr}(O P_{jk}) \) exhibit either constructive or destructive interference:
\[
\operatorname{tr}(O P_{jk}) \in \left\{ 0, \frac{2}{2^n} \right\} \quad \Rightarrow \quad \operatorname{tr}(O_0 P_{jk}) = \pm \frac{1}{2^n}.
\]
For example, consider \( |\phi\rangle = \frac{1}{\sqrt{2}}(|\ell_1\rangle + i |\ell_2\rangle) \). Among the projectors \( P_{jk} \), there are \( 2^{n-1} \) cases where the projection of \( U_j^\dagger |k\rangle \) onto the subspace spanned by \( \{|\ell_1\rangle, |\ell_2\rangle\} \) satisfies:
\[
(|\ell_1\rangle\langle \ell_1| + |\ell_2\rangle\langle \ell_2|) U_j^\dagger |k\rangle = \pm\frac{1}{\sqrt{2^n}} (|\ell_1\rangle + i |\ell_2\rangle),
\]
which yields constructive interference and gives \(\operatorname{tr}(O P_{jk}) = \frac{2}{2^n}\). For the remaining \( 2^{n-1} \) projectors, the projection becomes
\[
(|\ell_1\rangle\langle \ell_1| + |\ell_2\rangle\langle \ell_2|) U_j^\dagger |k\rangle =\pm \frac{1}{\sqrt{2^n}} (|\ell_1\rangle - i |\ell_2\rangle),
\]
leading to destructive interference and \(\operatorname{tr}(O P_{jk}) = 0\).  
  Therefore,
  \[
  \operatorname{tr}^2(O_0 P_{jk}) = \frac{1}{4^n}.
  \]

  \item \textbf{\( \mathcal{U}_{O,2} \)}: The relative phase differs from that of \( |\phi\rangle \), leading to cancellation of cross terms and uniform projection:
  \[
  \operatorname{tr}(O P_{jk}) = \frac{1}{2^n}, \quad \Rightarrow \quad \operatorname{tr}(O_0 P_{jk}) = 0.
  \]
Continuing the above example with \( |\phi\rangle = \frac{1}{\sqrt{2}}(|\ell_1\rangle + i |\ell_2\rangle) \), we observe that for all MUB states \( U_j^\dagger |k\rangle \) in this class, their projection onto the subspace spanned by \( \{|\ell_1\rangle, |\ell_2\rangle\} \) takes the form:
\[
(|\ell_1\rangle\langle \ell_1| + |\ell_2\rangle\langle \ell_2|) U_j^\dagger |k\rangle = \pm \frac{1}{\sqrt{2^n}} (|\ell_1\rangle \pm |\ell_2\rangle).
\]
Then $\operatorname{tr}(O P_{jk}) = \frac{1}{2^n}$. 
  
\end{itemize}

Due to the symmetric structure of MUB states expression, the two sets of unitaries \( \mathcal{U}_{O,1} \) and \( \mathcal{U}_{O,2} \) each contain exactly half of the nontrivial MUB unitaries, i.e., \( 2^{n-1} \) elements respectively.  
 Therefore, the total variance contribution from nontrivial MUBs is:
\[
2^{n+1} \cdot \frac{2^n }{2} \cdot \frac{1}{4^n} < 1.
\]

Combining both computational basis and notrivial MUB parts:
\[
\|O_0\|^2_{\text{shadow}} < \frac{1}{2} + 1 = \frac{3}{2}.
\]

Since the shadow norm under biased MUB measurements is roughly half that of random Clifford measurements, we deduce:
\[
N_{\mathrm{MUB}} \approx \frac{1}{2} N_{\mathrm{Clifford}}.
\]
\end{proof}

\begin{theorem}
To estimate \(K\) distinct off-diagonal density matrix elements \(\rho_{j_1k_1}, \dots, \rho_{j_Kk_K}\) with additive error at most \(\epsilon\), i.e., satisfying
\[
|\tilde{\rho}_{j_tk_t} - \rho_{j_tk_t}| \leq \epsilon \quad \text{for all } t = 1, \dots, K,
\]
the required sample complexity under the classical shadow tomography framework is as follows:

\begin{enumerate}
    \item \emph{Diagonal Elements:}  
    Perform projective measurements in the computational basis to estimate up to \(\min(2K, d) = \mathcal{O}(K)\) diagonal entries. The sample complexity for this step is
    \[
    O\left(\frac{\log K}{\epsilon^2}\right).
    \]
    
    \item \emph{Off-Diagonal Elements:}  
    Perform DMP-ST with random Clifford measurements or biased MUB measurements to estimate the rank-1 observables
    \[
    \left\{ |\phi_{j_tk_t}\rangle\langle \phi_{j_tk_t}|,\; |\psi_{j_tk_t}\rangle\langle \psi_{j_tk_t}| \;\middle|\; t = 1, \dots, K \right\}.
    \]
    The sample complexity for this step is also
    \[
    O\left(\frac{\log K}{\epsilon^2}\right).
    \]
\end{enumerate}
\end{theorem}
\begin{proof}
By Eqs.~(\ref{equ:r+i}) and~(\ref{equ:total}), each off-diagonal element \(\rho_{jk}\) can be recovered from a linear combination of three expectation values:
\[
\rho_{jk} = \rho^{*}_{kj}=\mbox{tr}(\rho(|\phi_{jk}^{+}\rangle\langle\phi_{jk}^{+}|-i|\psi_{jk}^{+}\rangle\langle\psi_{jk}^{+}|))-\frac{1-i}{2}(\rho_{kk}+\rho_{jj}). 
\]
To ensure the total estimation error is within \(\epsilon\), it suffices to estimate each component with error at most \(\epsilon/4\).

Since classical shadow tomography provides an additive-error estimator with sample complexity scaling as \(\mathcal{O}(\log M / \epsilon^2)\) for estimating \(M\) observables, applying it to the \(\mathcal{O}(K)\) diagonal terms and \(2K\) rank-1 observables yields an overall sample complexity of
\(
O\left(\frac{\log K}{\epsilon^2}\right)
\)
for each part. Combining both steps completes the proof.
\end{proof}

%\textbf{Measurement Configuration vs. Sample Complexity.}
\subsection{Comparing DMP and DMP-ST: Measurement Configuration and Sample Complexity}

In conventional direct measurement protocols (DMPs) \cite{lundeen2011direct,lundeen2012procedure,salvail2013full,bamber2014observing,shi2015scan,thekkadath2016direct,bolduc2016direct,vallone2016strong,calderaro2018direct,pan2019direct,zhang2020direct,xu2024resource}, estimating each off-diagonal matrix element \(\rho_{jk}\) typically requires a tailored measurement configuration—such as a specific ancilla interaction or interference setup. Consequently, estimating \(K\) matrix elements generally entails \(\mathcal{O}(K)=mK\) distinct device settings, with the proportionality constant $m$ depending on the specific experimental implementation. 
%Here we means the type of different unitary operations. 
When \(K\) is small (e.g., \(K=1\)), this requirement is modest, and it scales linearly with the number of target entries in the worst case.

By contrast, DMP-ST significantly improves the sample complexity, requiring only \(\mathcal{O}(\log(K/\delta)/\epsilon^2)\) samples of $\rho$ to estimate all \(K\) target elements. This is achieved by reusing a common set of randomized measurements for multiple estimations, rather than dedicating individual configurations to each observable. 

However, this sample efficiency comes with a trade-off in terms of measurement configurations. In the worst case, DMP-ST requires \(\mathcal{O}(\log(K/\delta)/\epsilon^2)\) distinct random unitary settings—each corresponding to a different measurement basis. This scaling arises from the need to gather sufficiently diverse measurement outcomes to ensure accurate reconstruction.

While such overhead may appear excessive, especially for small \(K\), it can be slightly mitigated through biased MUB measurements. Compared to uniformly random Clifford sampling for classical shadow tomography, biased MUB strategies not only offer a modest improvement in sample complexity—roughly by a factor of two by Lemma 2—but also reduce experimental burden in two key ways.

First, since the computational basis is selected with probability \(1/2\), the number of distinct device configurations is effectively halved, resulting in less frequent switching of measurement settings. Second, the circuit structure of each nontrivial MUB basis is simpler: it can be implemented using only $H$–$S$–$CZ$ gates \cite{yu2023efficient}, as opposed to the more complex 7-stage decomposition—$CX$–$CZ$–$S$–$H$–$S$–$CZ$–$CX$—typically required for Clifford-based measurements \cite{aaronson2004improved,koenig2014efficiently,maslov2018shorter,van2021simple,bravyi2021hadamard}.

These advantages make biased MUB sampling a more practical choice for DMP-ST in scenarios where measurement efficiency and experimental feasibility are critical considerations.

Briefly, when the number of target matrix elements \(K\) is small, $\mathcal{O}(\log(K/\delta)/\epsilon^2)\gg mK$, DMP-ST (by Clifford or biased MUB) may incur a higher measurement configuration cost than conventional DMPs. This is particularly evident under relaxed accuracy \(\epsilon\) or confidence \(\delta\) requirements, where traditional DMPs can achieve the same estimation with a fixed and at most $mK$ device settings.

However, this situation reverses for large \(K\). As the number of target entries increases, the logarithmic dependence of DMP-ST yields strict advantages in both sample and configuration complexity. For example, in full density matrix estimation where \(K = \mathcal{O}(d^2)\), traditional DMPs require at least \(\mathcal{O}(d)\) configurations, whereas DMP-ST achieves the same task using only \(\mathcal{O}(\log(d/\delta)/\epsilon^2)\) random measurement settings—an exponential reduction.

These exponential improvements in both sample complexity and measurement configuration requirements may seem surprisingly strong. However, it does not imply that \(\mathcal{O}(\log d/\epsilon^2)\) random measurement configurations are sufficient to perform full quantum state tomography for arbitrary \(d\)-dimensional quantum states. Precisely, this exponential gain applies to entrywise estimation, not to full quantum state tomography in trace distance. Accurately reconstructing the entire density matrix still requires additional considerations, which we address in the following section.

\section{Entrywise Estimation and Trace-Distance-Guaranteed Quantum State Tomography}

\subsection{Introduction and Trace Distance Motivation for Quantum State Tomography}

In this section, we demonstrate how classical shadow tomography can be adapted to reconstruct the full density matrix with rigorous guarantees in trace distance. We analyze both the sample complexity and classical computational complexity required to achieve this level of accuracy.

Understanding the sample complexity of quantum state tomography (QST) has long been a central question in quantum information theory, with applications in entanglement verification, quantum property testing, and quantum device calibration.

While classical shadow tomography protocols are primarily designed to estimate expectation values of the form \(\tr(\rho O)\) for a given set of observables $O_1, O_2, \dots$, the class of observables that can be efficiently handled is typically restricted. In particular, these protocols often assume that each observable is $k$-local—i.e., a tensor product involving at most $k$ nontrivial operators—or that it satisfies \(\tr(O^2) \le \operatorname{poly}(\log d)\). For more general observables that do not fall into these categories, one could resort to full quantum state tomography techniques to estimate $\tilde{\rho}$.

The goal of QST is to reconstruct the entire quantum state $\tilde{\rho}$ from measurement result such that the trace distance with target $\rho$ satisfies
\[
D_{\mathrm{tr}}(\rho, \tilde{\rho}) = \frac{1}{2} \| \rho - \tilde{\rho} \|_1 \le \epsilon.
\]
This ensures accurate estimation of \emph{all} observables with bounded operator norm, making trace distance a fundamental metric for QST. 

Some works use the infidelity \( \mathrm{IF}(\rho, \tilde{\rho}) = 1 - F(\rho, \tilde{\rho}) \) to characterize the estimation error between \( \rho \) and \( \tilde{\rho} \), rather than the trace distance. These two quantities are closely related via the inequalities \cite{fuchs1999cryptographic}:
\[
1 - F(\rho, \tilde{\rho}) \leq D_{\mathrm{tr}}(\rho, \tilde{\rho})  \leq \sqrt{1 - F(\rho, \tilde{\rho})^2},
\]
which ensures that bounds on the infidelity imply corresponding bounds on the trace distance, and vice versa.

\begin{remark}
To understand why trace distance provides a strong operational guarantee, consider the estimation error for any observable \(O\). Let \(\Delta = \tilde{\rho} - \rho\). Then,
\[
\left| \tr(\tilde{\rho} O) - \tr(\rho O) \right| = \left| \tr(\Delta O) \right| \le \|\Delta\|_1 \cdot \|O\|_\infty.
\]
Here, \(\|\Delta\|_1 = 2D_{\mathrm{tr}}(\rho, \tilde{\rho})\) quantifies distinguishability, and $\|O\|_\infty$ is the largest absolute eigenvalue of \(O\). 
Thus, small trace distance implies uniformly small errors for all bounded observables. 
For example, if $O$ is an arbitrary $n$-qubit Pauli observable, $\|O\|_\infty=1$ but $\tr(O^2)=2^n$.  
\end{remark}

\subsection{ Early Approaches and Sample Complexity Bounds}

A standard approach to quantum state tomography (QST) performs full Pauli measurements followed by linear inversion \cite{nielsen2010quantum}:
\[
\tilde{\rho} = \frac{1}{2^n} \sum_{P \in \mathcal{P}_n} \operatorname{tr}(\rho P) \cdot P,
\]
where \( \mathcal{P}_n \) is the set of \( n \)-qubit Pauli operators and \( d = 2^n \). This protocol requires estimating the expectation values of all \( d^2 \) observables, making it impractical for large systems.

To improve efficiency, later works exploit low-rank structure and apply compressed sensing techniques. Assuming that the density matrix \( \rho \) has rank \( r \ll d \), one can reconstruct it by solving a convex optimization problem using only \( \mathcal{O}(r d \log^2 d) \) Pauli expectations \cite{gross2010quantum}. To achieve trace distance error \( \epsilon \), the sample complexity for rank-\( r \) states using Pauli measurements is \( \mathcal{O}(r^2 d^2 \log d / \epsilon^2) \) \cite{flammia2012quantum}.

Further improvements can be obtained by using rank-one measurements drawn from random vectors. \citet{kueng2017low} showed that with Gaussian measurements, one can recover rank-\( r \) states with trace distance error \( \epsilon \) using \( \mathcal{O}(d r^2 / \epsilon^2) \) samples. For structured measurements from an approximate 4-design, the sample complexity becomes \( \mathcal{O}(d r^2 \log d / \epsilon^2) \), reflecting the limited randomness of the design. Notably, while \cite{kueng2017low} does not present results in terms of trace distance, this translation is clarified in the discussion following Theorem 1 of \cite{haah2017tit}.

When entangled (collective) measurements on \( \rho^{\otimes N} \) are allowed, the sample complexity can be further reduced to
\[
N = \mathcal{O}\left( \frac{d^2}{\varepsilon^2} \ln\left( \frac{d}{\varepsilon} \right) \right),
\]
as shown by \citet{o2016efficient} and \citet{haah2016sample}. This outperforms protocols based on single-copy measurements. The Holevo-type lower bound states that at least
\[
\Omega\left( \frac{dr}{\varepsilon^2 \log(d / r\varepsilon)} \right)
\]
copies are necessary to achieve trace distance error \( \varepsilon \) for any protocols, indicating that collective measurement schemes approach the information-theoretic limit. However, such entangled measurements are well beyond the reach of current NISQ-era quantum devices for the large $N$.

Limited to single-copy measurements on $\rho$, even when adaptive strategies are allowed—that is, when each new measurement can depend on previous measurement outcomes—\citet{chen2023does} proves a tight lower bound of $\Omega(d^3/\epsilon^2)$ for quantum state tomography. This matches the upper bound in \citet{kueng2017low}, confirming its optimality in the incoherent measurement setting. The proof constructs a hard distribution over quantum states via Gaussian ensembles and shows that any separable protocol fails to recover the state accurately unless the sample complexity meets this bound. While the result is information-theoretically tight, it is non-constructive and does not provide explicit measurement or reconstruction procedures.

These limitations motivate the development of shadow-based protocols that are both statistically and computationally efficient under the trace distance metric. We now present such a method.

\subsection{Our Approach: DMP-ST for Quantum State Tomography With Trace Distance Error}

 We introduce a tomography protocol based on the DMP-ST framework that achieves the same sample complexity of \(\mathcal{O}(d^3 \log d / \epsilon^2)\) as the approximate 4-design method in \citet{kueng2017low}, but using more practical measurements.

Instead of approximate 4-designs, we employ random Clifford circuits (forming a 3-design on $n$-qubit systems) or biased MUBs (forming a 2-design). These measurement schemes are implementable with \(\mathcal{O}(n^2)\) gates, significantly simplifying physical realization. Moreover, unlike the semidefinite program (SDP) required, our protocol estimates each matrix element directly, avoiding global convex optimization. In the approach of \citet{kueng2017low}, one needs to solve a semi-definite program (SDP) that minimizes the nuclear norm in order to find the optimal \(d\)-dimensional density matrix \(\tilde{\rho}\) consistent with the \(\mathcal{O}(d^2 \log d)\)  expectation values. 
In contrast, the DMP-ST method offers a more direct reconstruction: each matrix element can be estimated individually in a step-by-step fashion.

We formalize the performance guarantee below:

\begin{theorem}
Let \(\rho\) be a \(d \times d\) density matrix. To obtain an estimate \(\tilde{\rho}\) such that the trace distance satisfies:
\[
D_{\mathrm{tr}}(\rho, \tilde{\rho}) \leq \epsilon,
\]
we can use the DMP-ST with random Clifford measurements or biased MUB meaasurements to directly estimate each density matrix element \(\rho_{jk}\), ensuring that the element-wise error satisfies:
\[
|\tilde{\rho}_{jk} - \rho_{jk}| \leq \epsilon_1 \leq \frac{2\epsilon}{d^{3/2}} \quad \text{for all } j, k.
\]

The total sample complexity for achieving this trace distance error $\epsilon$ is:
\[
O\left(\frac{d^3 \log d}{\epsilon^2}\right).
\]
\end{theorem}

\begin{proof}
To ensure the trace distance \(D_{\mathrm{tr}}(\rho, \tilde{\rho})\) between \(\rho\) and \(\tilde{\rho}\) satisfies the desired bound, we first consider the trace distance definition:
\[
D_{\mathrm{tr}}(\rho, \tilde{\rho}) = \frac{1}{2} \|\rho - \tilde{\rho}\|_{1} = \frac{1}{2} \mathrm{tr}\left(\sqrt{(\rho - \tilde{\rho})^\dagger (\rho - \tilde{\rho})}\right).
\]
Let the eigenvalues of \((\rho - \tilde{\rho})^\dagger (\rho - \tilde{\rho})\) be \(\lambda_1, \lambda_2, \dots, \lambda_d\). Using the trace property:
\begin{equation}\label{eq:lambda}
   \lambda_1 + \lambda_2 + \cdots + \lambda_d = \mathrm{tr}((\rho - \tilde{\rho})^\dagger (\rho - \tilde{\rho})). 
\end{equation}
Assuming that \(|\tilde{\rho}_{jk} - \rho_{jk}| \leq \epsilon_1\) for all \(j, k\), it follows:
\[
\mathrm{tr}((\rho - \tilde{\rho})^\dagger (\rho - \tilde{\rho}))=\sum_{jk} |\tilde{\rho}_{jk} - \rho_{jk}|^2\leq d^2 \epsilon_1^2.
\]

Using the Cauchy-Schwarz inequality:
\[
\left(\sqrt{\lambda_1} + \cdots + \sqrt{\lambda_d}\right)^2 \leq d \left(\lambda_1 + \cdots + \lambda_d\right),
\]
we derive:
\[
D_{\mathrm{tr}}(\rho, \tilde{\rho}) = \left(\sqrt{\lambda_1} + \cdots + \sqrt{\lambda_d}\right)/2 \leq \frac{\sqrt{d}}{2} \sqrt{\lambda_1 + \cdots + \lambda_d}.
\]
Substituting the bound on \(\lambda_1 + \cdots + \lambda_d\) in Eq. (\ref{eq:lambda}):
\[
D_{\mathrm{tr}}(\rho, \tilde{\rho}) \leq \frac{\sqrt{d}}{2} \sqrt{d^2 \epsilon_1^2} = \frac{d^{3/2} \epsilon_1}{2}.
\]

To ensure \(D_{\mathrm{tr}}(\rho, \tilde{\rho}) \leq \epsilon\), the element-wise error should satisfy \(\epsilon_1 \leq \frac{2\epsilon}{d^{3/2}}\).    
Therefore, we should use DMP-ST to  estimate all the density matrix elements with element-wise error \(|\tilde{\rho}_{jk} - \rho_{jk}| \leq \epsilon_1\), where \(\epsilon_1 = \frac{2\epsilon}{d^{3/2}}\). The sampling complexity is \( O\left(\frac{\log d}{\epsilon_1^2}\right) \) for diagonal elements and \( O\left(\frac{\log (2d^2 -2 d)}{\epsilon_1^2}\right) \) for off-diagonal elements by Theorem 2. 
Combining both contributions and taking  \(\epsilon_1 \leq \frac{2\epsilon}{d^{3/2}}\), the total sampling complexity is:
\[
O\left(\frac{d^3 \log d}{\epsilon^2}\right).
\] 
\end{proof}

This matches the sample complexity of approximate 4-design-based schemes \cite{kueng2017low}, but our measurements and post-processing are simpler and better suited to implementation on near-term quantum devices. Moreover, it approaches the lower bound \(\Omega(d^3 / \epsilon^2)\) for any adaptive or nonadaptive single-copy tomography \cite{chen2023does}. 

\subsection{Computational Complexity Analysis}

The DMP-ST method not only achieves near-optimal sample complexity but also features different classical post-processing. 

In the MUB-based implementation, each observable corresponds to a rank-1 projector with only two nonzero amplitudes. 
In each post-processing of classical shadow tomography, we should calculate $\tr(O\mathcal{M}^{-1}(U_j^{\dagger}|k\rangle\langle k|U_j))$, where $U_j^{\dagger}|k\rangle$ is one projected MUB state. 
Thus, evaluating each estimation  requires only \( \mathcal{O}(\log d) \) operations to access and compute the relevant MUB state amplitudes.

Since we estimate \( \mathcal{O}(d^2) \) matrix elements, and each sample contributes to a single entry with \( \mathcal{O}(\log d) \) cost, the total classical processing over all \( \mathcal{O}(d^3/\epsilon^2) \) samples is:
\[
O\left( \frac{d^5 (\log d)^2}{\epsilon^2} \right).
\]

It is worth highlighting that the post-processing strategy of DMP-ST fundamentally differs from all  traditional quantum state tomography (QST) protocols designed to achieve trace distance accuracy \(\epsilon\). Previous QST approaches typically perform repeated measurements over a fixed informationally complete POVM \(\{E_k\}_{k=1}^{L(d)}\) \cite{flammia2005minimal}, with \(L(d)\) minimized using prior knowledge of quantum state \cite{heinosaari2013quantum}, and estimate the corresponding probabilities \(\{p_k \approx \tr(E_k \rho)\}\)  or related expectation values. The density matrix \(\rho\) is then reconstructed using techniques such as linear inversion, semidefinite programming (SDP), maximum likelihood estimation, or projected least squares—each seeking a best-fit quantum state $\tilde{\rho}$ consistent with the observed data \cite{paris2004quantum}.

Even when the special QST method attains the optimal sample complexity of \(\mathcal{O}(d^3 / \epsilon^2)\), computing the empirical probability distribution \(\{p_k\}\) already incurs \(\mathcal{O}(d^3 / \epsilon^2)\) computational time. This is followed by additional overhead from the reconstruction step to recover the $d\times d$ density matrix $\tilde{\rho}$, which typically involves global optimization procedures with substantial classical cost.

\section{Conclusion}
The direct estimation of multiple entries of a quantum state’s density matrix is a fundamental and practically significant task in quantum information science. However, traditional DMPs often suffer from poor scalability and measurement inefficiency, particularly in high-dimensional settings.

Unlike all previously known direct measurement protocols, the proposed DMP-ST framework enables global estimation: a single batch of measurement outcomes can be reused to estimate arbitrary matrix elements. In contrast, traditional DMPs—including weak measurement schemes, phase-shifting interferometry, and basis-encoded projection methods—typically require a distinct measurement configuration for each new off-diagonal element, and often discard a large portion of data due to post-selection or filtering.

These limitations hinder scalability: estimating \(K\) matrix elements generally demands fresh quantum samples and newly configured measurement settings for each target. DMP-ST overcomes these issues by leveraging the classical shadow framework to uniformly estimate all target entries from a shared dataset. As a result, the sample complexity for estimating \(K\) off-diagonal elements reduces from \(\mathcal{O}(K / \epsilon^2)\) to \(\mathcal{O}(\log K / \epsilon^2)\), with a corresponding exponential reduction in the number of required measurement settings for $K$ large enough.

A key advantage of DMP-ST is its reusability: once a sufficiently large dataset is collected to ensure accuracy \(\epsilon\) and confidence level \(1 - \delta\), no further measurements are required to estimate any additional off-diagonal matrix elements. 

We further show that estimating all \(d^2\) matrix entries to additive accuracy \(\epsilon_1 = \mathcal{O}(\epsilon / d^{3/2})\) guarantees trace distance error \(D_{\mathrm{tr}}(\rho, \tilde{\rho}) \le \epsilon\), with total sample complexity \(\mathcal{O}(d^3 \log d / \epsilon^2)\), closely matching the theoretical lower bound for single-copy measurements for quantum state tomography.

In addition, we analyze biased MUB measurements within the classical shadow framework and prove that, for the observables used in DMP-ST, the shadow norm remains bounded by 3/2. As a result, the sample complexity is reduced to one-half of that required by random Clifford-based classical shadow tomography method. 

In summary, DMP-ST with randomly Clifford measurements or biased MUB measurements provides a scalable and near-optimal framework for reconstructing individual density matrix elements and performing full quantum state tomography, bridging rigorous theoretical guarantees with practical efficiency in high-dimensional quantum systems.

\textbf{Acknowledgements---}
This work received support from the National Natural Science Foundation of China through Grants No. 62001260 and No. 42330707, and from the Beijing Natural Science Foundation under Grant No. Z220002.   %


\begin{thebibliography}{47}%
\makeatletter
\providecommand \@ifxundefined [1]{%
 \@ifx{#1\undefined}
}%
\providecommand \@ifnum [1]{%
 \ifnum #1\expandafter \@firstoftwo
 \else \expandafter \@secondoftwo
 \fi
}%
\providecommand \@ifx [1]{%
 \ifx #1\expandafter \@firstoftwo
 \else \expandafter \@secondoftwo
 \fi
}%
\providecommand \natexlab [1]{#1}%
\providecommand \enquote  [1]{``#1''}%
\providecommand \bibnamefont  [1]{#1}%
\providecommand \bibfnamefont [1]{#1}%
\providecommand \citenamefont [1]{#1}%
\providecommand \href@noop [0]{\@secondoftwo}%
\providecommand \href [0]{\begingroup \@sanitize@url \@href}%
\providecommand \@href[1]{\@@startlink{#1}\@@href}%
\providecommand \@@href[1]{\endgroup#1\@@endlink}%
\providecommand \@sanitize@url [0]{\catcode `\\12\catcode `\$12\catcode
  `\&12\catcode `\#12\catcode `\^12\catcode `\_12\catcode `\%12\relax}%
\providecommand \@@startlink[1]{}%
\providecommand \@@endlink[0]{}%
\providecommand \url  [0]{\begingroup\@sanitize@url \@url }%
\providecommand \@url [1]{\endgroup\@href {#1}{\urlprefix }}%
\providecommand \urlprefix  [0]{URL }%
\providecommand \Eprint [0]{\href }%
\providecommand \doibase [0]{http://dx.doi.org/}%
\providecommand \selectlanguage [0]{\@gobble}%
\providecommand \bibinfo  [0]{\@secondoftwo}%
\providecommand \bibfield  [0]{\@secondoftwo}%
\providecommand \translation [1]{[#1]}%
\providecommand \BibitemOpen [0]{}%
\providecommand \bibitemStop [0]{}%
\providecommand \bibitemNoStop [0]{.\EOS\space}%
\providecommand \EOS [0]{\spacefactor3000\relax}%
\providecommand \BibitemShut  [1]{\csname bibitem#1\endcsname}%
\let\auto@bib@innerbib\@empty
%</preamble>
\bibitem [{\citenamefont {Paris}\ and\ \citenamefont
  {Rehacek}(2004)}]{paris2004quantum}%
  \BibitemOpen
  \bibfield  {author} {\bibinfo {author} {\bibfnamefont {M.}~\bibnamefont
  {Paris}}\ and\ \bibinfo {author} {\bibfnamefont {J.}~\bibnamefont
  {Rehacek}},\ }\href@noop {} {\emph {\bibinfo {title} {Quantum state
  estimation}}},\ Vol.\ \bibinfo {volume} {649}\ (\bibinfo  {publisher}
  {Springer Science \& Business Media},\ \bibinfo {year} {2004})\BibitemShut
  {NoStop}%
\bibitem [{\citenamefont {Nielsen}\ and\ \citenamefont
  {Chuang}(2010)}]{nielsen2010quantum}%
  \BibitemOpen
  \bibfield  {author} {\bibinfo {author} {\bibfnamefont {M.~A.}\ \bibnamefont
  {Nielsen}}\ and\ \bibinfo {author} {\bibfnamefont {I.~L.}\ \bibnamefont
  {Chuang}},\ }\href@noop {} {\emph {\bibinfo {title} {Quantum computation and
  quantum information}}}\ (\bibinfo  {publisher} {Cambridge university press},\
  \bibinfo {year} {2010})\BibitemShut {NoStop}%
\bibitem [{\citenamefont {Friis}\ \emph {et~al.}(2019)\citenamefont {Friis},
  \citenamefont {Vitagliano}, \citenamefont {Malik},\ and\ \citenamefont
  {Huber}}]{friis2019entanglement}%
  \BibitemOpen
  \bibfield  {author} {\bibinfo {author} {\bibfnamefont {N.}~\bibnamefont
  {Friis}}, \bibinfo {author} {\bibfnamefont {G.}~\bibnamefont {Vitagliano}},
  \bibinfo {author} {\bibfnamefont {M.}~\bibnamefont {Malik}}, \ and\ \bibinfo
  {author} {\bibfnamefont {M.}~\bibnamefont {Huber}},\ }\href@noop {}
  {\bibfield  {journal} {\bibinfo  {journal} {Nature Reviews Physics}\ }\textbf
  {\bibinfo {volume} {1}},\ \bibinfo {pages} {72} (\bibinfo {year}
  {2019})}\BibitemShut {NoStop}%
\bibitem [{\citenamefont {Streltsov}\ \emph {et~al.}(2017)\citenamefont
  {Streltsov}, \citenamefont {Adesso},\ and\ \citenamefont
  {Plenio}}]{streltsov2017colloquium}%
  \BibitemOpen
  \bibfield  {author} {\bibinfo {author} {\bibfnamefont {A.}~\bibnamefont
  {Streltsov}}, \bibinfo {author} {\bibfnamefont {G.}~\bibnamefont {Adesso}}, \
  and\ \bibinfo {author} {\bibfnamefont {M.~B.}\ \bibnamefont {Plenio}},\
  }\href@noop {} {\bibfield  {journal} {\bibinfo  {journal} {Reviews of Modern
  Physics}\ }\textbf {\bibinfo {volume} {89}},\ \bibinfo {pages} {041003}
  (\bibinfo {year} {2017})}\BibitemShut {NoStop}%
\bibitem [{\citenamefont {Gross}\ \emph {et~al.}(2010)\citenamefont {Gross},
  \citenamefont {Liu}, \citenamefont {Flammia}, \citenamefont {Becker},\ and\
  \citenamefont {Eisert}}]{gross2010quantum}%
  \BibitemOpen
  \bibfield  {author} {\bibinfo {author} {\bibfnamefont {D.}~\bibnamefont
  {Gross}}, \bibinfo {author} {\bibfnamefont {Y.-K.}\ \bibnamefont {Liu}},
  \bibinfo {author} {\bibfnamefont {S.~T.}\ \bibnamefont {Flammia}}, \bibinfo
  {author} {\bibfnamefont {S.}~\bibnamefont {Becker}}, \ and\ \bibinfo {author}
  {\bibfnamefont {J.}~\bibnamefont {Eisert}},\ }\href@noop {} {\bibfield
  {journal} {\bibinfo  {journal} {Physical review letters}\ }\textbf {\bibinfo
  {volume} {105}},\ \bibinfo {pages} {150401} (\bibinfo {year}
  {2010})}\BibitemShut {NoStop}%
\bibitem [{\citenamefont {Flammia}\ \emph {et~al.}(2012)\citenamefont
  {Flammia}, \citenamefont {Gross}, \citenamefont {Liu},\ and\ \citenamefont
  {Eisert}}]{flammia2012quantum}%
  \BibitemOpen
  \bibfield  {author} {\bibinfo {author} {\bibfnamefont {S.~T.}\ \bibnamefont
  {Flammia}}, \bibinfo {author} {\bibfnamefont {D.}~\bibnamefont {Gross}},
  \bibinfo {author} {\bibfnamefont {Y.-K.}\ \bibnamefont {Liu}}, \ and\
  \bibinfo {author} {\bibfnamefont {J.}~\bibnamefont {Eisert}},\ }\href@noop {}
  {\bibfield  {journal} {\bibinfo  {journal} {New Journal of Physics}\ }\textbf
  {\bibinfo {volume} {14}},\ \bibinfo {pages} {095022} (\bibinfo {year}
  {2012})}\BibitemShut {NoStop}%
\bibitem [{\citenamefont {Lundeen}\ \emph {et~al.}(2011)\citenamefont
  {Lundeen}, \citenamefont {Sutherland}, \citenamefont {Patel}, \citenamefont
  {Stewart},\ and\ \citenamefont {Bamber}}]{lundeen2011direct}%
  \BibitemOpen
  \bibfield  {author} {\bibinfo {author} {\bibfnamefont {J.~S.}\ \bibnamefont
  {Lundeen}}, \bibinfo {author} {\bibfnamefont {B.}~\bibnamefont {Sutherland}},
  \bibinfo {author} {\bibfnamefont {A.}~\bibnamefont {Patel}}, \bibinfo
  {author} {\bibfnamefont {C.}~\bibnamefont {Stewart}}, \ and\ \bibinfo
  {author} {\bibfnamefont {C.}~\bibnamefont {Bamber}},\ }\href@noop {}
  {\bibfield  {journal} {\bibinfo  {journal} {Nature}\ }\textbf {\bibinfo
  {volume} {474}},\ \bibinfo {pages} {188} (\bibinfo {year}
  {2011})}\BibitemShut {NoStop}%
\bibitem [{\citenamefont {Lundeen}\ and\ \citenamefont
  {Bamber}(2012)}]{lundeen2012procedure}%
  \BibitemOpen
  \bibfield  {author} {\bibinfo {author} {\bibfnamefont {J.~S.}\ \bibnamefont
  {Lundeen}}\ and\ \bibinfo {author} {\bibfnamefont {C.}~\bibnamefont
  {Bamber}},\ }\href@noop {} {\bibfield  {journal} {\bibinfo  {journal}
  {Physical Review Letters}\ }\textbf {\bibinfo {volume} {108}},\ \bibinfo
  {pages} {070402} (\bibinfo {year} {2012})}\BibitemShut {NoStop}%
\bibitem [{\citenamefont {Salvail}\ \emph {et~al.}(2013)\citenamefont
  {Salvail}, \citenamefont {Agnew}, \citenamefont {Johnson}, \citenamefont
  {Bolduc}, \citenamefont {Leach},\ and\ \citenamefont
  {Boyd}}]{salvail2013full}%
  \BibitemOpen
  \bibfield  {author} {\bibinfo {author} {\bibfnamefont {J.~Z.}\ \bibnamefont
  {Salvail}}, \bibinfo {author} {\bibfnamefont {M.}~\bibnamefont {Agnew}},
  \bibinfo {author} {\bibfnamefont {A.~S.}\ \bibnamefont {Johnson}}, \bibinfo
  {author} {\bibfnamefont {E.}~\bibnamefont {Bolduc}}, \bibinfo {author}
  {\bibfnamefont {J.}~\bibnamefont {Leach}}, \ and\ \bibinfo {author}
  {\bibfnamefont {R.~W.}\ \bibnamefont {Boyd}},\ }\href@noop {} {\bibfield
  {journal} {\bibinfo  {journal} {Nature Photonics}\ }\textbf {\bibinfo
  {volume} {7}},\ \bibinfo {pages} {316} (\bibinfo {year} {2013})}\BibitemShut
  {NoStop}%
\bibitem [{\citenamefont {Bamber}\ and\ \citenamefont
  {Lundeen}(2014)}]{bamber2014observing}%
  \BibitemOpen
  \bibfield  {author} {\bibinfo {author} {\bibfnamefont {C.}~\bibnamefont
  {Bamber}}\ and\ \bibinfo {author} {\bibfnamefont {J.~S.}\ \bibnamefont
  {Lundeen}},\ }\href@noop {} {\bibfield  {journal} {\bibinfo  {journal}
  {Physical Review Letters}\ }\textbf {\bibinfo {volume} {112}},\ \bibinfo
  {pages} {070405} (\bibinfo {year} {2014})}\BibitemShut {NoStop}%
\bibitem [{\citenamefont {Shi}\ \emph {et~al.}(2015)\citenamefont {Shi},
  \citenamefont {Mirhosseini}, \citenamefont {Margiewicz}, \citenamefont
  {Malik}, \citenamefont {Rivera}, \citenamefont {Zhu},\ and\ \citenamefont
  {Boyd}}]{shi2015scan}%
  \BibitemOpen
  \bibfield  {author} {\bibinfo {author} {\bibfnamefont {Z.}~\bibnamefont
  {Shi}}, \bibinfo {author} {\bibfnamefont {M.}~\bibnamefont {Mirhosseini}},
  \bibinfo {author} {\bibfnamefont {J.}~\bibnamefont {Margiewicz}}, \bibinfo
  {author} {\bibfnamefont {M.}~\bibnamefont {Malik}}, \bibinfo {author}
  {\bibfnamefont {F.}~\bibnamefont {Rivera}}, \bibinfo {author} {\bibfnamefont
  {Z.}~\bibnamefont {Zhu}}, \ and\ \bibinfo {author} {\bibfnamefont {R.~W.}\
  \bibnamefont {Boyd}},\ }\href@noop {} {\bibfield  {journal} {\bibinfo
  {journal} {Optica}\ }\textbf {\bibinfo {volume} {2}},\ \bibinfo {pages} {388}
  (\bibinfo {year} {2015})}\BibitemShut {NoStop}%
\bibitem [{\citenamefont {Thekkadath}\ \emph {et~al.}(2016)\citenamefont
  {Thekkadath}, \citenamefont {Giner}, \citenamefont {Chalich}, \citenamefont
  {Horton}, \citenamefont {Banker},\ and\ \citenamefont
  {Lundeen}}]{thekkadath2016direct}%
  \BibitemOpen
  \bibfield  {author} {\bibinfo {author} {\bibfnamefont {G.~S.}\ \bibnamefont
  {Thekkadath}}, \bibinfo {author} {\bibfnamefont {L.}~\bibnamefont {Giner}},
  \bibinfo {author} {\bibfnamefont {Y.}~\bibnamefont {Chalich}}, \bibinfo
  {author} {\bibfnamefont {M.~J.}\ \bibnamefont {Horton}}, \bibinfo {author}
  {\bibfnamefont {J.}~\bibnamefont {Banker}}, \ and\ \bibinfo {author}
  {\bibfnamefont {J.~S.}\ \bibnamefont {Lundeen}},\ }\href@noop {} {\bibfield
  {journal} {\bibinfo  {journal} {Physical Review Letters}\ }\textbf {\bibinfo
  {volume} {117}},\ \bibinfo {pages} {120401} (\bibinfo {year}
  {2016})}\BibitemShut {NoStop}%
\bibitem [{\citenamefont {Bolduc}\ \emph {et~al.}(2016)\citenamefont {Bolduc},
  \citenamefont {Gariepy},\ and\ \citenamefont {Leach}}]{bolduc2016direct}%
  \BibitemOpen
  \bibfield  {author} {\bibinfo {author} {\bibfnamefont {E.}~\bibnamefont
  {Bolduc}}, \bibinfo {author} {\bibfnamefont {G.}~\bibnamefont {Gariepy}}, \
  and\ \bibinfo {author} {\bibfnamefont {J.}~\bibnamefont {Leach}},\
  }\href@noop {} {\bibfield  {journal} {\bibinfo  {journal} {Nature
  Communications}\ }\textbf {\bibinfo {volume} {7}},\ \bibinfo {pages} {1}
  (\bibinfo {year} {2016})}\BibitemShut {NoStop}%
\bibitem [{\citenamefont {Vallone}\ and\ \citenamefont
  {Dequal}(2016)}]{vallone2016strong}%
  \BibitemOpen
  \bibfield  {author} {\bibinfo {author} {\bibfnamefont {G.}~\bibnamefont
  {Vallone}}\ and\ \bibinfo {author} {\bibfnamefont {D.}~\bibnamefont
  {Dequal}},\ }\href@noop {} {\bibfield  {journal} {\bibinfo  {journal}
  {Physical review letters}\ }\textbf {\bibinfo {volume} {116}},\ \bibinfo
  {pages} {040502} (\bibinfo {year} {2016})}\BibitemShut {NoStop}%
\bibitem [{\citenamefont {Calderaro}\ \emph {et~al.}(2018)\citenamefont
  {Calderaro}, \citenamefont {Foletto}, \citenamefont {Dequal}, \citenamefont
  {Villoresi},\ and\ \citenamefont {Vallone}}]{calderaro2018direct}%
  \BibitemOpen
  \bibfield  {author} {\bibinfo {author} {\bibfnamefont {L.}~\bibnamefont
  {Calderaro}}, \bibinfo {author} {\bibfnamefont {G.}~\bibnamefont {Foletto}},
  \bibinfo {author} {\bibfnamefont {D.}~\bibnamefont {Dequal}}, \bibinfo
  {author} {\bibfnamefont {P.}~\bibnamefont {Villoresi}}, \ and\ \bibinfo
  {author} {\bibfnamefont {G.}~\bibnamefont {Vallone}},\ }\href@noop {}
  {\bibfield  {journal} {\bibinfo  {journal} {Physical Review Letters}\
  }\textbf {\bibinfo {volume} {121}},\ \bibinfo {pages} {230501} (\bibinfo
  {year} {2018})}\BibitemShut {NoStop}%
\bibitem [{\citenamefont {Pan}\ \emph {et~al.}(2019)\citenamefont {Pan},
  \citenamefont {Xu}, \citenamefont {Kedem}, \citenamefont {Wang},
  \citenamefont {Chen}, \citenamefont {Jan}, \citenamefont {Sun}, \citenamefont
  {Xu}, \citenamefont {Han}, \citenamefont {Li} \emph
  {et~al.}}]{pan2019direct}%
  \BibitemOpen
  \bibfield  {author} {\bibinfo {author} {\bibfnamefont {W.-W.}\ \bibnamefont
  {Pan}}, \bibinfo {author} {\bibfnamefont {X.-Y.}\ \bibnamefont {Xu}},
  \bibinfo {author} {\bibfnamefont {Y.}~\bibnamefont {Kedem}}, \bibinfo
  {author} {\bibfnamefont {Q.-Q.}\ \bibnamefont {Wang}}, \bibinfo {author}
  {\bibfnamefont {Z.}~\bibnamefont {Chen}}, \bibinfo {author} {\bibfnamefont
  {M.}~\bibnamefont {Jan}}, \bibinfo {author} {\bibfnamefont {K.}~\bibnamefont
  {Sun}}, \bibinfo {author} {\bibfnamefont {J.-S.}\ \bibnamefont {Xu}},
  \bibinfo {author} {\bibfnamefont {Y.-J.}\ \bibnamefont {Han}}, \bibinfo
  {author} {\bibfnamefont {C.-F.}\ \bibnamefont {Li}},  \emph {et~al.},\
  }\href@noop {} {\bibfield  {journal} {\bibinfo  {journal} {Physical Review
  Letters}\ }\textbf {\bibinfo {volume} {123}},\ \bibinfo {pages} {150402}
  (\bibinfo {year} {2019})}\BibitemShut {NoStop}%
\bibitem [{\citenamefont {Zhang}\ \emph {et~al.}(2020)\citenamefont {Zhang},
  \citenamefont {Hu}, \citenamefont {Hou}, \citenamefont {Tang}, \citenamefont
  {Zhu}, \citenamefont {Xiang}, \citenamefont {Li}, \citenamefont {Guo},\ and\
  \citenamefont {Zhang}}]{zhang2020direct}%
  \BibitemOpen
  \bibfield  {author} {\bibinfo {author} {\bibfnamefont {C.-R.}\ \bibnamefont
  {Zhang}}, \bibinfo {author} {\bibfnamefont {M.-J.}\ \bibnamefont {Hu}},
  \bibinfo {author} {\bibfnamefont {Z.-B.}\ \bibnamefont {Hou}}, \bibinfo
  {author} {\bibfnamefont {J.-F.}\ \bibnamefont {Tang}}, \bibinfo {author}
  {\bibfnamefont {J.}~\bibnamefont {Zhu}}, \bibinfo {author} {\bibfnamefont
  {G.-Y.}\ \bibnamefont {Xiang}}, \bibinfo {author} {\bibfnamefont {C.-F.}\
  \bibnamefont {Li}}, \bibinfo {author} {\bibfnamefont {G.-C.}\ \bibnamefont
  {Guo}}, \ and\ \bibinfo {author} {\bibfnamefont {Y.-S.}\ \bibnamefont
  {Zhang}},\ }\href@noop {} {\bibfield  {journal} {\bibinfo  {journal}
  {Physical Review A}\ }\textbf {\bibinfo {volume} {101}},\ \bibinfo {pages}
  {012119} (\bibinfo {year} {2020})}\BibitemShut {NoStop}%
\bibitem [{\citenamefont {Xu}\ \emph {et~al.}(2024)\citenamefont {Xu},
  \citenamefont {Zhou}, \citenamefont {Tao}, \citenamefont {Zhong},
  \citenamefont {Wang}, \citenamefont {Cao}, \citenamefont {Xia}, \citenamefont
  {Wang}, \citenamefont {Zhan}, \citenamefont {Zhang} \emph
  {et~al.}}]{xu2024resource}%
  \BibitemOpen
  \bibfield  {author} {\bibinfo {author} {\bibfnamefont {L.}~\bibnamefont
  {Xu}}, \bibinfo {author} {\bibfnamefont {M.}~\bibnamefont {Zhou}}, \bibinfo
  {author} {\bibfnamefont {R.}~\bibnamefont {Tao}}, \bibinfo {author}
  {\bibfnamefont {Z.}~\bibnamefont {Zhong}}, \bibinfo {author} {\bibfnamefont
  {B.}~\bibnamefont {Wang}}, \bibinfo {author} {\bibfnamefont {Z.}~\bibnamefont
  {Cao}}, \bibinfo {author} {\bibfnamefont {H.}~\bibnamefont {Xia}}, \bibinfo
  {author} {\bibfnamefont {Q.}~\bibnamefont {Wang}}, \bibinfo {author}
  {\bibfnamefont {H.}~\bibnamefont {Zhan}}, \bibinfo {author} {\bibfnamefont
  {A.}~\bibnamefont {Zhang}},  \emph {et~al.},\ }\href@noop {} {\bibfield
  {journal} {\bibinfo  {journal} {Physical Review Letters}\ }\textbf {\bibinfo
  {volume} {132}},\ \bibinfo {pages} {030201} (\bibinfo {year}
  {2024})}\BibitemShut {NoStop}%
\bibitem [{\citenamefont {Huang}\ \emph {et~al.}(2020)\citenamefont {Huang},
  \citenamefont {Kueng},\ and\ \citenamefont {Preskill}}]{huang2020predicting}%
  \BibitemOpen
  \bibfield  {author} {\bibinfo {author} {\bibfnamefont {H.-Y.}\ \bibnamefont
  {Huang}}, \bibinfo {author} {\bibfnamefont {R.}~\bibnamefont {Kueng}}, \ and\
  \bibinfo {author} {\bibfnamefont {J.}~\bibnamefont {Preskill}},\ }\href@noop
  {} {\bibfield  {journal} {\bibinfo  {journal} {Nature Physics}\ }\textbf
  {\bibinfo {volume} {16}},\ \bibinfo {pages} {1050} (\bibinfo {year}
  {2020})}\BibitemShut {NoStop}%
\bibitem [{\citenamefont {Aaronson}(2018)}]{aaronson2018shadow}%
  \BibitemOpen
  \bibfield  {author} {\bibinfo {author} {\bibfnamefont {S.}~\bibnamefont
  {Aaronson}},\ }in\ \href@noop {} {\emph {\bibinfo {booktitle} {Proceedings of
  the 50th annual ACM SIGACT symposium on theory of computing}}}\ (\bibinfo
  {year} {2018})\ pp.\ \bibinfo {pages} {325--338}\BibitemShut {NoStop}%
\bibitem [{\citenamefont {Kokail}\ \emph {et~al.}(2021)\citenamefont {Kokail},
  \citenamefont {van Bijnen}, \citenamefont {Elben}, \citenamefont
  {Vermersch},\ and\ \citenamefont {Zoller}}]{kokail2021entanglement}%
  \BibitemOpen
  \bibfield  {author} {\bibinfo {author} {\bibfnamefont {C.}~\bibnamefont
  {Kokail}}, \bibinfo {author} {\bibfnamefont {R.}~\bibnamefont {van Bijnen}},
  \bibinfo {author} {\bibfnamefont {A.}~\bibnamefont {Elben}}, \bibinfo
  {author} {\bibfnamefont {B.}~\bibnamefont {Vermersch}}, \ and\ \bibinfo
  {author} {\bibfnamefont {P.}~\bibnamefont {Zoller}},\ }\href@noop {}
  {\bibfield  {journal} {\bibinfo  {journal} {Nature Physics}\ }\textbf
  {\bibinfo {volume} {17}},\ \bibinfo {pages} {936} (\bibinfo {year}
  {2021})}\BibitemShut {NoStop}%
\bibitem [{\citenamefont {Hadfield}\ \emph {et~al.}(2022)\citenamefont
  {Hadfield}, \citenamefont {Bravyi}, \citenamefont {Raymond},\ and\
  \citenamefont {Mezzacapo}}]{hadfield2022measurements}%
  \BibitemOpen
  \bibfield  {author} {\bibinfo {author} {\bibfnamefont {C.}~\bibnamefont
  {Hadfield}}, \bibinfo {author} {\bibfnamefont {S.}~\bibnamefont {Bravyi}},
  \bibinfo {author} {\bibfnamefont {R.}~\bibnamefont {Raymond}}, \ and\
  \bibinfo {author} {\bibfnamefont {A.}~\bibnamefont {Mezzacapo}},\ }\href@noop
  {} {\bibfield  {journal} {\bibinfo  {journal} {Communications in Mathematical
  Physics}\ }\textbf {\bibinfo {volume} {391}},\ \bibinfo {pages} {951}
  (\bibinfo {year} {2022})}\BibitemShut {NoStop}%
\bibitem [{\citenamefont {Seif}\ \emph {et~al.}(2023)\citenamefont {Seif},
  \citenamefont {Cian}, \citenamefont {Zhou}, \citenamefont {Chen},\ and\
  \citenamefont {Jiang}}]{seif2023shadow}%
  \BibitemOpen
  \bibfield  {author} {\bibinfo {author} {\bibfnamefont {A.}~\bibnamefont
  {Seif}}, \bibinfo {author} {\bibfnamefont {Z.-P.}\ \bibnamefont {Cian}},
  \bibinfo {author} {\bibfnamefont {S.}~\bibnamefont {Zhou}}, \bibinfo {author}
  {\bibfnamefont {S.}~\bibnamefont {Chen}}, \ and\ \bibinfo {author}
  {\bibfnamefont {L.}~\bibnamefont {Jiang}},\ }\href@noop {} {\bibfield
  {journal} {\bibinfo  {journal} {PRX Quantum}\ }\textbf {\bibinfo {volume}
  {4}},\ \bibinfo {pages} {010303} (\bibinfo {year} {2023})}\BibitemShut
  {NoStop}%
\bibitem [{\citenamefont {Jnane}\ \emph {et~al.}(2024)\citenamefont {Jnane},
  \citenamefont {Steinberg}, \citenamefont {Cai}, \citenamefont {Nguyen},\ and\
  \citenamefont {Koczor}}]{jnane2024quantum}%
  \BibitemOpen
  \bibfield  {author} {\bibinfo {author} {\bibfnamefont {H.}~\bibnamefont
  {Jnane}}, \bibinfo {author} {\bibfnamefont {J.}~\bibnamefont {Steinberg}},
  \bibinfo {author} {\bibfnamefont {Z.}~\bibnamefont {Cai}}, \bibinfo {author}
  {\bibfnamefont {H.~C.}\ \bibnamefont {Nguyen}}, \ and\ \bibinfo {author}
  {\bibfnamefont {B.}~\bibnamefont {Koczor}},\ }\href@noop {} {\bibfield
  {journal} {\bibinfo  {journal} {PRX Quantum}\ }\textbf {\bibinfo {volume}
  {5}},\ \bibinfo {pages} {010324} (\bibinfo {year} {2024})}\BibitemShut
  {NoStop}%
\bibitem [{\citenamefont {Zhao}\ and\ \citenamefont
  {Miyake}(2024)}]{zhao2024group}%
  \BibitemOpen
  \bibfield  {author} {\bibinfo {author} {\bibfnamefont {A.}~\bibnamefont
  {Zhao}}\ and\ \bibinfo {author} {\bibfnamefont {A.}~\bibnamefont {Miyake}},\
  }\href@noop {} {\bibfield  {journal} {\bibinfo  {journal} {npj Quantum
  Information}\ }\textbf {\bibinfo {volume} {10}},\ \bibinfo {pages} {57}
  (\bibinfo {year} {2024})}\BibitemShut {NoStop}%
\bibitem [{\citenamefont {Huang}\ \emph {et~al.}(2022)\citenamefont {Huang},
  \citenamefont {Kueng}, \citenamefont {Torlai}, \citenamefont {Albert},\ and\
  \citenamefont {Preskill}}]{huang2022provably}%
  \BibitemOpen
  \bibfield  {author} {\bibinfo {author} {\bibfnamefont {H.-Y.}\ \bibnamefont
  {Huang}}, \bibinfo {author} {\bibfnamefont {R.}~\bibnamefont {Kueng}},
  \bibinfo {author} {\bibfnamefont {G.}~\bibnamefont {Torlai}}, \bibinfo
  {author} {\bibfnamefont {V.~V.}\ \bibnamefont {Albert}}, \ and\ \bibinfo
  {author} {\bibfnamefont {J.}~\bibnamefont {Preskill}},\ }\href@noop {}
  {\bibfield  {journal} {\bibinfo  {journal} {Science}\ }\textbf {\bibinfo
  {volume} {377}},\ \bibinfo {pages} {eabk3333} (\bibinfo {year}
  {2022})}\BibitemShut {NoStop}%
\bibitem [{\citenamefont {Jerbi}\ \emph {et~al.}(2024)\citenamefont {Jerbi},
  \citenamefont {Gyurik}, \citenamefont {Marshall}, \citenamefont {Molteni},\
  and\ \citenamefont {Dunjko}}]{jerbi2024shadows}%
  \BibitemOpen
  \bibfield  {author} {\bibinfo {author} {\bibfnamefont {S.}~\bibnamefont
  {Jerbi}}, \bibinfo {author} {\bibfnamefont {C.}~\bibnamefont {Gyurik}},
  \bibinfo {author} {\bibfnamefont {S.~C.}\ \bibnamefont {Marshall}}, \bibinfo
  {author} {\bibfnamefont {R.}~\bibnamefont {Molteni}}, \ and\ \bibinfo
  {author} {\bibfnamefont {V.}~\bibnamefont {Dunjko}},\ }\href@noop {}
  {\bibfield  {journal} {\bibinfo  {journal} {Nature Communications}\ }\textbf
  {\bibinfo {volume} {15}},\ \bibinfo {pages} {5676} (\bibinfo {year}
  {2024})}\BibitemShut {NoStop}%
\bibitem [{\citenamefont {Kueng}\ \emph {et~al.}(2017)\citenamefont {Kueng},
  \citenamefont {Rauhut},\ and\ \citenamefont {Terstiege}}]{kueng2017low}%
  \BibitemOpen
  \bibfield  {author} {\bibinfo {author} {\bibfnamefont {R.}~\bibnamefont
  {Kueng}}, \bibinfo {author} {\bibfnamefont {H.}~\bibnamefont {Rauhut}}, \
  and\ \bibinfo {author} {\bibfnamefont {U.}~\bibnamefont {Terstiege}},\
  }\href@noop {} {\bibfield  {journal} {\bibinfo  {journal} {Applied and
  Computational Harmonic Analysis}\ }\textbf {\bibinfo {volume} {42}},\
  \bibinfo {pages} {88} (\bibinfo {year} {2017})}\BibitemShut {NoStop}%
\bibitem [{\citenamefont {Chen}\ \emph {et~al.}(2023)\citenamefont {Chen},
  \citenamefont {Huang}, \citenamefont {Li}, \citenamefont {Liu},\ and\
  \citenamefont {Sellke}}]{chen2023does}%
  \BibitemOpen
  \bibfield  {author} {\bibinfo {author} {\bibfnamefont {S.}~\bibnamefont
  {Chen}}, \bibinfo {author} {\bibfnamefont {B.}~\bibnamefont {Huang}},
  \bibinfo {author} {\bibfnamefont {J.}~\bibnamefont {Li}}, \bibinfo {author}
  {\bibfnamefont {A.}~\bibnamefont {Liu}}, \ and\ \bibinfo {author}
  {\bibfnamefont {M.}~\bibnamefont {Sellke}},\ }in\ \href@noop {} {\emph
  {\bibinfo {booktitle} {2023 IEEE 64th Annual Symposium on Foundations of
  Computer Science (FOCS)}}}\ (\bibinfo {organization} {IEEE},\ \bibinfo {year}
  {2023})\ pp.\ \bibinfo {pages} {391--404}\BibitemShut {NoStop}%
\bibitem [{\citenamefont {Horodecki}\ \emph {et~al.}(2009)\citenamefont
  {Horodecki}, \citenamefont {Horodecki}, \citenamefont {Horodecki},\ and\
  \citenamefont {Horodecki}}]{horodecki2009quantum}%
  \BibitemOpen
  \bibfield  {author} {\bibinfo {author} {\bibfnamefont {R.}~\bibnamefont
  {Horodecki}}, \bibinfo {author} {\bibfnamefont {P.}~\bibnamefont
  {Horodecki}}, \bibinfo {author} {\bibfnamefont {M.}~\bibnamefont
  {Horodecki}}, \ and\ \bibinfo {author} {\bibfnamefont {K.}~\bibnamefont
  {Horodecki}},\ }\href@noop {} {\bibfield  {journal} {\bibinfo  {journal}
  {Reviews of modern physics}\ }\textbf {\bibinfo {volume} {81}},\ \bibinfo
  {pages} {865} (\bibinfo {year} {2009})}\BibitemShut {NoStop}%
\bibitem [{\citenamefont {Ringbauer}\ \emph {et~al.}(2018)\citenamefont
  {Ringbauer}, \citenamefont {Bromley}, \citenamefont {Cianciaruso},
  \citenamefont {Lami}, \citenamefont {Lau}, \citenamefont {Adesso},
  \citenamefont {White}, \citenamefont {Fedrizzi},\ and\ \citenamefont
  {Piani}}]{ringbauer2018certification}%
  \BibitemOpen
  \bibfield  {author} {\bibinfo {author} {\bibfnamefont {M.}~\bibnamefont
  {Ringbauer}}, \bibinfo {author} {\bibfnamefont {T.~R.}\ \bibnamefont
  {Bromley}}, \bibinfo {author} {\bibfnamefont {M.}~\bibnamefont
  {Cianciaruso}}, \bibinfo {author} {\bibfnamefont {L.}~\bibnamefont {Lami}},
  \bibinfo {author} {\bibfnamefont {W.~S.}\ \bibnamefont {Lau}}, \bibinfo
  {author} {\bibfnamefont {G.}~\bibnamefont {Adesso}}, \bibinfo {author}
  {\bibfnamefont {A.~G.}\ \bibnamefont {White}}, \bibinfo {author}
  {\bibfnamefont {A.}~\bibnamefont {Fedrizzi}}, \ and\ \bibinfo {author}
  {\bibfnamefont {M.}~\bibnamefont {Piani}},\ }\href@noop {} {\bibfield
  {journal} {\bibinfo  {journal} {Physical Review X}\ }\textbf {\bibinfo
  {volume} {8}},\ \bibinfo {pages} {041007} (\bibinfo {year}
  {2018})}\BibitemShut {NoStop}%
\bibitem [{\citenamefont {Feng}\ \emph {et~al.}(2021)\citenamefont {Feng},
  \citenamefont {Ren},\ and\ \citenamefont {Zhou}}]{feng2021direct}%
  \BibitemOpen
  \bibfield  {author} {\bibinfo {author} {\bibfnamefont {T.}~\bibnamefont
  {Feng}}, \bibinfo {author} {\bibfnamefont {C.}~\bibnamefont {Ren}}, \ and\
  \bibinfo {author} {\bibfnamefont {X.}~\bibnamefont {Zhou}},\ }\href@noop {}
  {\bibfield  {journal} {\bibinfo  {journal} {Physical Review A}\ }\textbf
  {\bibinfo {volume} {104}},\ \bibinfo {pages} {042403} (\bibinfo {year}
  {2021})}\BibitemShut {NoStop}%
\bibitem [{\citenamefont {Wang}\ \emph {et~al.}(2024)\citenamefont {Wang},
  \citenamefont {Jiang}, \citenamefont {Liu},\ and\ \citenamefont
  {Li}}]{wang2024direct}%
  \BibitemOpen
  \bibfield  {author} {\bibinfo {author} {\bibfnamefont {Y.}~\bibnamefont
  {Wang}}, \bibinfo {author} {\bibfnamefont {H.}~\bibnamefont {Jiang}},
  \bibinfo {author} {\bibfnamefont {Y.}~\bibnamefont {Liu}}, \ and\ \bibinfo
  {author} {\bibfnamefont {K.}~\bibnamefont {Li}},\ }\href@noop {} {\bibfield
  {journal} {\bibinfo  {journal} {arXiv preprint arXiv:2409.03435}\ } (\bibinfo
  {year} {2024})}\BibitemShut {NoStop}%
\bibitem [{\citenamefont {Caves}\ \emph {et~al.}(2002)\citenamefont {Caves},
  \citenamefont {Fuchs},\ and\ \citenamefont {Schack}}]{Caves2002}%
  \BibitemOpen
  \bibfield  {author} {\bibinfo {author} {\bibfnamefont {C.~M.}\ \bibnamefont
  {Caves}}, \bibinfo {author} {\bibfnamefont {C.~A.}\ \bibnamefont {Fuchs}}, \
  and\ \bibinfo {author} {\bibfnamefont {R.}~\bibnamefont {Schack}},\
  }\href@noop {} {\bibfield  {journal} {\bibinfo  {journal} {Journal of
  Mathematical Physics}\ }\textbf {\bibinfo {volume} {43}},\ \bibinfo {pages}
  {4537} (\bibinfo {year} {2002})}\BibitemShut {NoStop}%
\bibitem [{\citenamefont {Wang}\ and\ \citenamefont
  {Cui}(2024)}]{wang2024classical}%
  \BibitemOpen
  \bibfield  {author} {\bibinfo {author} {\bibfnamefont {Y.}~\bibnamefont
  {Wang}}\ and\ \bibinfo {author} {\bibfnamefont {W.}~\bibnamefont {Cui}},\
  }\href@noop {} {\bibfield  {journal} {\bibinfo  {journal} {Physical Review
  A}\ }\textbf {\bibinfo {volume} {109}},\ \bibinfo {pages} {062406} (\bibinfo
  {year} {2024})}\BibitemShut {NoStop}%
\bibitem [{\citenamefont {Yu}\ and\ \citenamefont
  {Dongsheng}(2023)}]{yu2023efficient}%
  \BibitemOpen
  \bibfield  {author} {\bibinfo {author} {\bibfnamefont {W.}~\bibnamefont
  {Yu}}\ and\ \bibinfo {author} {\bibfnamefont {W.}~\bibnamefont {Dongsheng}},\
  }\href@noop {} {\bibfield  {journal} {\bibinfo  {journal} {arXiv preprint
  arXiv:2311.11698}\ } (\bibinfo {year} {2023})}\BibitemShut {NoStop}%
\bibitem [{\citenamefont {Aaronson}\ and\ \citenamefont
  {Gottesman}(2004)}]{aaronson2004improved}%
  \BibitemOpen
  \bibfield  {author} {\bibinfo {author} {\bibfnamefont {S.}~\bibnamefont
  {Aaronson}}\ and\ \bibinfo {author} {\bibfnamefont {D.}~\bibnamefont
  {Gottesman}},\ }\href@noop {} {\bibfield  {journal} {\bibinfo  {journal}
  {Physical Review A}\ }\textbf {\bibinfo {volume} {70}},\ \bibinfo {pages}
  {052328} (\bibinfo {year} {2004})}\BibitemShut {NoStop}%
\bibitem [{\citenamefont {Koenig}\ and\ \citenamefont
  {Smolin}(2014)}]{koenig2014efficiently}%
  \BibitemOpen
  \bibfield  {author} {\bibinfo {author} {\bibfnamefont {R.}~\bibnamefont
  {Koenig}}\ and\ \bibinfo {author} {\bibfnamefont {J.~A.}\ \bibnamefont
  {Smolin}},\ }\href@noop {} {\bibfield  {journal} {\bibinfo  {journal}
  {Journal of Mathematical Physics}\ }\textbf {\bibinfo {volume} {55}},\
  \bibinfo {pages} {122202} (\bibinfo {year} {2014})}\BibitemShut {NoStop}%
\bibitem [{\citenamefont {Maslov}\ and\ \citenamefont
  {Roetteler}(2018)}]{maslov2018shorter}%
  \BibitemOpen
  \bibfield  {author} {\bibinfo {author} {\bibfnamefont {D.}~\bibnamefont
  {Maslov}}\ and\ \bibinfo {author} {\bibfnamefont {M.}~\bibnamefont
  {Roetteler}},\ }\href@noop {} {\bibfield  {journal} {\bibinfo  {journal}
  {IEEE Transactions on Information Theory}\ }\textbf {\bibinfo {volume}
  {64}},\ \bibinfo {pages} {4729} (\bibinfo {year} {2018})}\BibitemShut
  {NoStop}%
\bibitem [{\citenamefont {van~den Berg}(2021)}]{van2021simple}%
  \BibitemOpen
  \bibfield  {author} {\bibinfo {author} {\bibfnamefont {E.}~\bibnamefont
  {van~den Berg}},\ }in\ \href@noop {} {\emph {\bibinfo {booktitle} {2021 IEEE
  International Conference on Quantum Computing and Engineering (QCE)}}}\
  (\bibinfo {year} {2021})\ pp.\ \bibinfo {pages} {54--59}\BibitemShut
  {NoStop}%
\bibitem [{\citenamefont {Bravyi}\ and\ \citenamefont
  {Maslov}(2021)}]{bravyi2021hadamard}%
  \BibitemOpen
  \bibfield  {author} {\bibinfo {author} {\bibfnamefont {S.}~\bibnamefont
  {Bravyi}}\ and\ \bibinfo {author} {\bibfnamefont {D.}~\bibnamefont
  {Maslov}},\ }\href@noop {} {\bibfield  {journal} {\bibinfo  {journal} {IEEE
  Transactions on Information Theory}\ }\textbf {\bibinfo {volume} {67}},\
  \bibinfo {pages} {4546} (\bibinfo {year} {2021})}\BibitemShut {NoStop}%
\bibitem [{\citenamefont {Fuchs}\ and\ \citenamefont {Van
  De~Graaf}(1999)}]{fuchs1999cryptographic}%
  \BibitemOpen
  \bibfield  {author} {\bibinfo {author} {\bibfnamefont {C.~A.}\ \bibnamefont
  {Fuchs}}\ and\ \bibinfo {author} {\bibfnamefont {J.}~\bibnamefont {Van
  De~Graaf}},\ }\href@noop {} {\bibfield  {journal} {\bibinfo  {journal} {IEEE
  Transactions on Information Theory}\ }\textbf {\bibinfo {volume} {45}},\
  \bibinfo {pages} {1216} (\bibinfo {year} {1999})}\BibitemShut {NoStop}%
\bibitem [{\citenamefont {Haah}\ \emph {et~al.}(2017)\citenamefont {Haah},
  \citenamefont {Harrow}, \citenamefont {Ji}, \citenamefont {Wu},\ and\
  \citenamefont {Yu}}]{haah2017tit}%
  \BibitemOpen
  \bibfield  {author} {\bibinfo {author} {\bibfnamefont {J.}~\bibnamefont
  {Haah}}, \bibinfo {author} {\bibfnamefont {A.~W.}\ \bibnamefont {Harrow}},
  \bibinfo {author} {\bibfnamefont {Z.}~\bibnamefont {Ji}}, \bibinfo {author}
  {\bibfnamefont {X.}~\bibnamefont {Wu}}, \ and\ \bibinfo {author}
  {\bibfnamefont {N.}~\bibnamefont {Yu}},\ }\href@noop {} {\bibfield  {journal}
  {\bibinfo  {journal} {IEEE Transactions on Information Theory}\ }\textbf
  {\bibinfo {volume} {63}},\ \bibinfo {pages} {5628} (\bibinfo {year}
  {2017})}\BibitemShut {NoStop}%
\bibitem [{\citenamefont {O'Donnell}\ and\ \citenamefont
  {Wright}(2016)}]{o2016efficient}%
  \BibitemOpen
  \bibfield  {author} {\bibinfo {author} {\bibfnamefont {R.}~\bibnamefont
  {O'Donnell}}\ and\ \bibinfo {author} {\bibfnamefont {J.}~\bibnamefont
  {Wright}},\ }in\ \href@noop {} {\emph {\bibinfo {booktitle} {Proceedings of
  the forty-eighth annual ACM symposium on Theory of Computing (STOC)}}}\
  (\bibinfo {year} {2016})\ pp.\ \bibinfo {pages} {899--912}\BibitemShut
  {NoStop}%
\bibitem [{\citenamefont {Haah}\ \emph {et~al.}(2016)\citenamefont {Haah},
  \citenamefont {Harrow}, \citenamefont {Ji}, \citenamefont {Wu},\ and\
  \citenamefont {Yu}}]{haah2016sample}%
  \BibitemOpen
  \bibfield  {author} {\bibinfo {author} {\bibfnamefont {J.}~\bibnamefont
  {Haah}}, \bibinfo {author} {\bibfnamefont {A.~W.}\ \bibnamefont {Harrow}},
  \bibinfo {author} {\bibfnamefont {Z.}~\bibnamefont {Ji}}, \bibinfo {author}
  {\bibfnamefont {X.}~\bibnamefont {Wu}}, \ and\ \bibinfo {author}
  {\bibfnamefont {N.}~\bibnamefont {Yu}},\ }in\ \href@noop {} {\emph {\bibinfo
  {booktitle} {Proceedings of the forty-eighth annual ACM symposium on Theory
  of Computing}}}\ (\bibinfo {year} {2016})\ pp.\ \bibinfo {pages}
  {913--925}\BibitemShut {NoStop}%
\bibitem [{\citenamefont {Flammia}\ \emph {et~al.}(2005)\citenamefont
  {Flammia}, \citenamefont {Silberfarb},\ and\ \citenamefont
  {Caves}}]{flammia2005minimal}%
  \BibitemOpen
  \bibfield  {author} {\bibinfo {author} {\bibfnamefont {S.~T.}\ \bibnamefont
  {Flammia}}, \bibinfo {author} {\bibfnamefont {A.}~\bibnamefont {Silberfarb}},
  \ and\ \bibinfo {author} {\bibfnamefont {C.~M.}\ \bibnamefont {Caves}},\
  }\href@noop {} {\bibfield  {journal} {\bibinfo  {journal} {Foundations of
  Physics}\ }\textbf {\bibinfo {volume} {35}},\ \bibinfo {pages} {1985}
  (\bibinfo {year} {2005})}\BibitemShut {NoStop}%
\bibitem [{\citenamefont {Heinosaari}\ \emph {et~al.}(2013)\citenamefont
  {Heinosaari}, \citenamefont {Mazzarella},\ and\ \citenamefont
  {Wolf}}]{heinosaari2013quantum}%
  \BibitemOpen
  \bibfield  {author} {\bibinfo {author} {\bibfnamefont {T.}~\bibnamefont
  {Heinosaari}}, \bibinfo {author} {\bibfnamefont {L.}~\bibnamefont
  {Mazzarella}}, \ and\ \bibinfo {author} {\bibfnamefont {M.~M.}\ \bibnamefont
  {Wolf}},\ }\href@noop {} {\bibfield  {journal} {\bibinfo  {journal}
  {Communications in Mathematical Physics}\ }\textbf {\bibinfo {volume}
  {318}},\ \bibinfo {pages} {355} (\bibinfo {year} {2013})}\BibitemShut
  {NoStop}%
\end{thebibliography}
\end{document}